\documentclass[12pt]{article}

\usepackage{amsmath, amssymb, amsthm}
\usepackage{enumerate}
\usepackage{multirow}
\usepackage{color}
\usepackage{graphicx}
\usepackage{natbib}
\usepackage{url}

\newcommand{\blind}{0}

\newcommand{\bmu}{\boldsymbol\mu}

\newcommand{\bP}{\mathbf{P}}
\newcommand{\bE}{\mathbf{E}}
\newcommand{\bV}{\mathbf{Var}}
\newcommand{\bCov}{\mathbf{Cov}}

\newcommand{\bPB}{\mathbf{P_B}}
\newcommand{\bEB}{\mathbf{E_B}}
\newcommand{\bVB}{\mathbf{Var_B}}
\newcommand{\bCovB}{\mathbf{Cov_B}}

\newtheorem{theorem}{Theorem}[section]

\newtheorem{corollary}[theorem]{Corollary}

\newtheorem{remark}[theorem]{Remark}
\newtheorem{assumption}{Assumption}[section]

\begin{document} 

\def\spacingset#1{\renewcommand{\baselinestretch}{#1}\small\normalsize} \spacingset{1}

\if0\blind
{
  \title{\bf A weighted edge-count two-sample test for multivariate and object data}
  \author{Hao Chen, Xu Chen, and Yi Su\thanks{
    Hao Chen is an Assistant Professor at the Department of Statistics at University of California, Davis. Xu Chen is a Master's student at the Department of Statistics at Duke University.  Yi Su is a PhD student at the Department of Statistics at University of California, Davis.}
\hspace{.2cm}\\
    }
\date{}
  \maketitle
} \fi

\if1\blind
{
  \bigskip
  \bigskip
  \bigskip
  \begin{center}
    {\LARGE\bf A weighted edge-count two-sample test for multivariate and object data}
\end{center}
  \medskip
} \fi

\bigskip
\begin{abstract}

Two-sample tests for multivariate data and non-Euclidean data are widely used in many fields.  Parametric tests are mostly restrained to certain types of data that meets the assumptions of the parametric models.  In this paper, we study a nonparametric testing procedure that utilizes graphs representing the similarity among observations.  It can be applied to any data types as long as an informative similarity measure on the sample space can be defined.  The classic test based on a similarity graph has a problem when the two sample sizes are different.  We solve the problem by applying appropriate weights to different components of the classic test statistic.  The new test exhibits substantial power gains in simulation studies.  Its asymptotic permutation null distribution is derived and shown to work well under finite samples, facilitating its application to large datasets.  The new test is illustrated through an analysis on a real dataset of network data.

\end{abstract}

\vfill
 
\newpage \spacingset{1.45} \section{Introduction}
\label{sec:intro}

Two-sample testing is a fundamental problem in statistics.  Due to the increasing richness of data in both dimension and complexity, this problem is encountering new challenges. Nowadays, many applications involve the test on data in high dimensions \citep{de2011clinical,feigenson2014disorganization} or even on non-Euclidean data, also referred to as object data, such as image data and network data \citep{eagle2009inferring,kossinets2006empirical}.  Parametric approaches can be applied to multivariate data under certain assumptions while their power decreases quickly as the dimension grows unless strong assumptions are made to facilitate the estimation of the large number of (nuisance) parameters, such as the covariance matrix.  In this work, we study a nonparametric testing procedure that works for both multivariate data and object data.

Nonparametric testing for two-sample differences has a long history and rich literature; see \cite{gibbons2011nonparametric} for a survey.  \cite{friedman1979multivariate} proposed the first practical test that can be applied to data with arbitrary dimension.  They constructed a minimum spanning tree (MST) on the pooled observations from both samples, which is a graph that connect all observations with the sum of distances of pairs of observations connected in the graph minimized.  They proposed to use the number of edges in the MST that connect nodes (observations) from different samples to be the test statistic.
We call this test the \emph{edge-count test} for easy reference.

The edge-count test is not restricted to the MST.  It can be applied to any similarity graph where more similar observations are more likely to be connected. \cite{friedman1979multivariate} also considered $k$-MSTs.  A $k$-MST is the union of the 1st, \dots, and $k$th MST, where the $i$th MST is a spanning tree connecting all observations that minimizes the sum of distances across edges subject to the constraint that this spanning tree does not contain any edge in the 1st, \dots, ($i-1$)th MST(s).  The edge-count test in general has a higher power when the graph is slighter denser, e.g. from 1-MST to 3-MST \citep{friedman1979multivariate}.  The edge-count test was also applied to $k$-nearest neighbor graphs \citep{schilling1986multivariate, henze1988multivariate}, where each observation connects to its first $k$ nearest neighbors, and minimum distance non-bipartite paired graph (MDP) \citep{rosenbaum2005exact}, where observations are divided into non-overlapping pairs in a way such that the sum of the distances of the pairs is minimized (when the total number of observations is odd, a pseudo observation, which has distance 0 to all observations, is added).  The MDP can be extended to $k$-MDPs as well, where a $k$-MDP is defined similarly to a $k$-MST.  Besides these common ways to construct the similarity graph, the graph can also be provided by domain experts directly \citep{chen2013graph}.

The rationale of the edge-count test is that, if the two samples are from different distributions, observations from the same sample are more similar and thus more likely to be connected in the similarity graph than observations from different samples.  The edge-count test rejects the null hypothesis of equal distribution if there is a significantly smaller number of between-sample edges than what is expected under the null.
\cite{maa1996reducing} showed that the edge-count test based on the MST constructed on Euclidean distance is consistent against all alternatives for multivariate data.

However, in practice, sample sizes range from tens to thousands, or somewhat larger.  We call these sample sizes \emph{practical sample sizes}.  \cite{chen2016new} found that, when the dimension of the data is moderate to high, for practical sample sizes, the edge-count test is effective for locational alternatives but can have low power for scale alternatives when the Euclidean distance is used to construct the similarity graph.  The authors proposed a new test statistic on the similarity graph that works for both locational and scale alternatives under practical sample sizes.  We call this test the \emph{generalized edge-count test} for easy reference.  

\cite{chen2016new} recommended to use the generalized edge-count test when there is no clue on the type of alternatives.  However, when the alternative is locational, the edge-count test is recommended as it in general has higher power than the generalized edge-count test under such circumstances. 

In this work, we addressed another problem for the edge-count test when the sample sizes of the two samples are different.  Taking the edge-count test on the 5-MST constructed on Euclidean distances for testing the mean difference of two 20-dimensional Gaussian distributions as an example, we found that, starting from the equal sample size scenario, the power of the edge-count test \emph{decreases} when one sample size is doubled and the other sample size keeps the same (see in Section \ref{sec:problem} for more details).  This is counter-intuitive as increasing the sample size adds more information and we would expect the power of the test to increase.  This weird phenomenon indicates that the edge-count test statistic is not well-defined. 

To address this problem, we propose a modified version of the edge-count test.  The idea is that, when the sample sizes are different, it is more difficult to form an edge within the sample with a smaller sample size than that for the sample with a larger sample size.  So we give within-sample edges different weights according to which sample they are from instead of treating them equally in the edge-count test.  The new test works properly under unequal sample sizes and we call it the \emph{weighted edge-count test}.

When the sample sizes are different, under locational alternatives, the weighted edge-count test is more powerful than the edge-count test and the generalized edge-count test. 
Hence, the weighted edge-count test and the generalized edge-count test can be used complementarily with one mainly for locational alternatives and the other for more general alternatives.  The proposed weighted edge-count test is implemented in an \texttt{R} package \texttt{gTests}.

The rest of the paper is organized as follows.  The problem of the edge-count test under unequal sample sizes is explored in Section \ref{sec:problem}.  The weighted edge-count test is proposed and studied in Section \ref{sec:weighted}.  Its power is examined in Section \ref{sec:power} and its asymptotics are studied in Section \ref{sec:asym}.  We illustrate the weighted edge-count test through an analysis on a real dataset of network data in Section \ref{sec:application}.  In Section \ref{sec:relation}, we discuss in more detail the relation between the weighted edge-count test and the generalized edge-count test.

\section{The problem}
\label{sec:problem}

In this section, we explore what happens to the edge-count test under unequal sample sizes.  As an illustration example, we randomly draw $m$ observations from distribution $F_1=\mathcal{N}(\mathbf{0}, I_d)$ and $n$ observations from distribution $F_2=\mathcal{N}(\bmu,I_d)$, $\|\bmu\|_2=1.3$, $d=50$, and call them Sample 1 and Sample 2, respectively.  Here, $\mathcal{N}(\bmu, \Sigma)$ denotes a multivariate normal distribution with mean $\bmu$ and covariance matrix $\Sigma$, and $I_d$ is the $d\times d$ identity matrix.  We use the common notation $\|\cdot\|_2$ to denote the $L_2$ norm.  The two distributions are chosen such that the test has moderate power.  We consider the following two scenarios:
\begin{itemize}
\item Scenario 1: $m=n=50$.
\item Scenario 2: $m=50,\ n=100$.
\end{itemize}
So for scenario 2, one sample size keeps the same and the other sample size is doubled.  We applied the edge-count test on the $k$-MST constructed on the Euclidean distance to the simulated data.  Intuitively, the test should have a higher power in scenario 2 than in scenario 1 because there are more observations in scenario 2.  We estimated the power of the test by the fraction of trials that the test rejected the null hypothesis at 0.05 significance level in 1,000 trials.  The results are shown in Figure \ref{fig:poweredge}.   

\begin{figure}[!htp]
\begin{center}
\includegraphics[width=.55\textwidth]{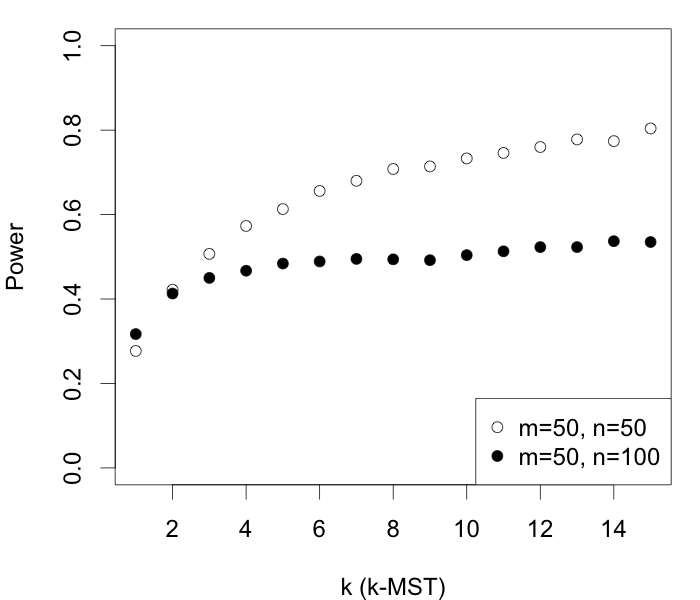}
\end{center}
\caption{The fraction of the trials (out of 1,000) that the edge-count test on $k$-MST rejected the null hypothesis of equal distribution.}\label{fig:poweredge}
\end{figure}

The choice of an optimal $k$ is an open question.  We here show the results from 1-MST to 15-MST.  To our surprise, we see from Figure \ref{fig:poweredge} that the edge-count test has lower power in scenario 2 under all $k$-MSTs except for 1-MST.  Why does this happen?

Before exploring in details, we first introduce some notations.  Let $N=m+n$ be the total sample size.  We pool observations from both samples and index them by $1,\dots, N$.  Let $G$ be an undirected similarity graph on the pooled observations (nodes) with no multi-edge, i.e., there is at most one edge between any two nodes.  The graph can be a $k$-MST, a $k$-MDP, etc.  We use $G$ to refer to both the graph and its set of edges when the vertex set is implicitly obvious.  The symbol $|\cdot|$ is used to denote the size of the set, so $|G|$ is the number of edges in $G$.  Let $R$ be the number of edges in $G$ that connect observations between the two samples, i.e., the number of between-sample edges.  We work under the permutation null distribution, which places $1/\binom{N}{m}$ probability on each of the $\binom{N}{m}$ permutations of the sample labels.  When there is no further specification, we denote by $\bP$, $\bE$, $\bV$, $\bCov$ probability, expectation, variance, and covariance, respectively, under the permutation null distribution. 

We next explore in details on what happens in the edge-count test.  We focus on the test based on the 5-MST (similar patterns are observed for other $k$-MSTs where $k>1$).  First of all, we check whether adding more observations does make $R$ further smaller than its null expectation $\bE(R)$. 
In each run, we calculate $\bE(R)-R$ for scenario 1 and denote it by $D_1$.  We then add 50 more observations randomly drawn from $F_2$, re-construct the 5-MST based on the 150 observations, calculate $\bE(R)-R$ based on the new 5-MST, and denote it by $D_2$.  We check whether $D_2$ is larger than $D_1$ in general.  Figure \ref{fig:Ddiff} shows the boxplots of $D_1$ and $D_2$ separately, as well as their differences $(D_2-D_1)$, from 1,000 simulation runs.  

\begin{figure}[!htp]
\begin{center}
\includegraphics[width=.48\textwidth]{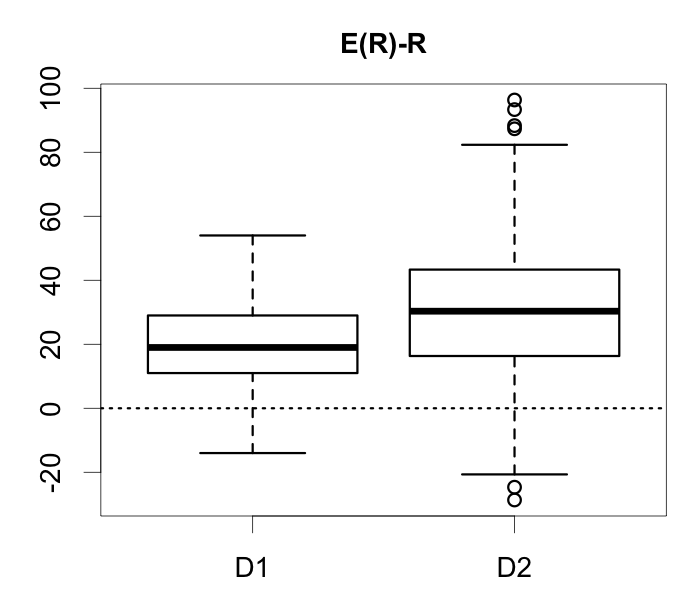}
\includegraphics[width=.48\textwidth]{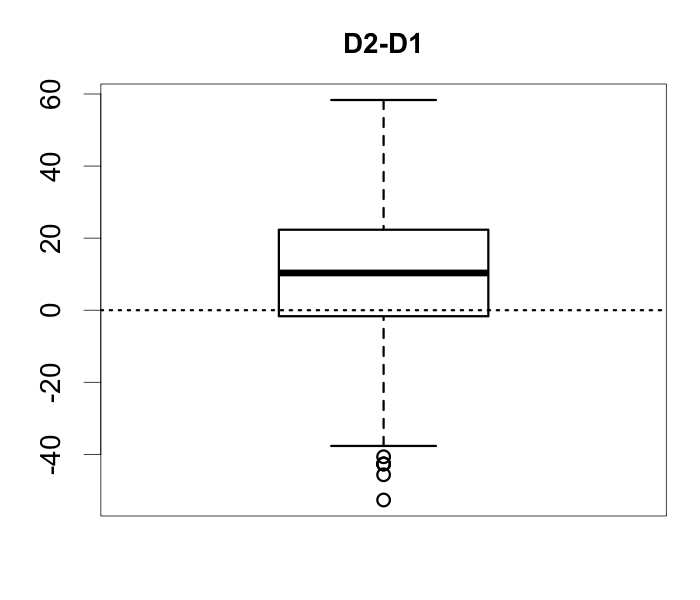}
\end{center}
\caption{Boxplots of $\bE(R)-R$ before and after adding 50 more observations, as well as their difference (after $-$ before), from 1,000 simulation runs.  The horizontal dashed line is at level 0.}\label{fig:Ddiff}
\end{figure}

We see that $\bE(R)-R$ is in general positive under both scenarios.  This is expected since the two distributions differ in the mean and the observations tend not to find observations from the other sample to be similar.  The boxplot of the difference between $D_2$ and $D_1$ shows that $\bE(R)-R$ indeed becomes larger in general when more observations are included.  This also complies with what we would expect.

To quantify how further $R$ is from its null expectation $\bE(R)$, we need to compare $\bE(R)-R$ to its standard deviation under the null.
Figure \ref{fig:sdR} shows the boxplots of the standard deviations of $R$ before ($sd_1$) and after ($sd_2$) adding the 50 more observations, as well as their ratio ($sd_2/sd_1$), from 1,000 simulation runs. 

\begin{figure}[!htp]
\begin{center}
\includegraphics[width=.48\textwidth]{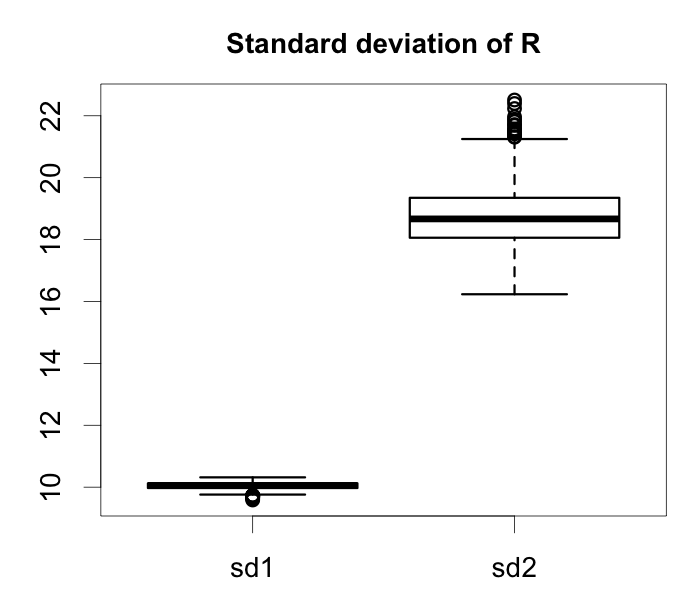}
\includegraphics[width=.48\textwidth]{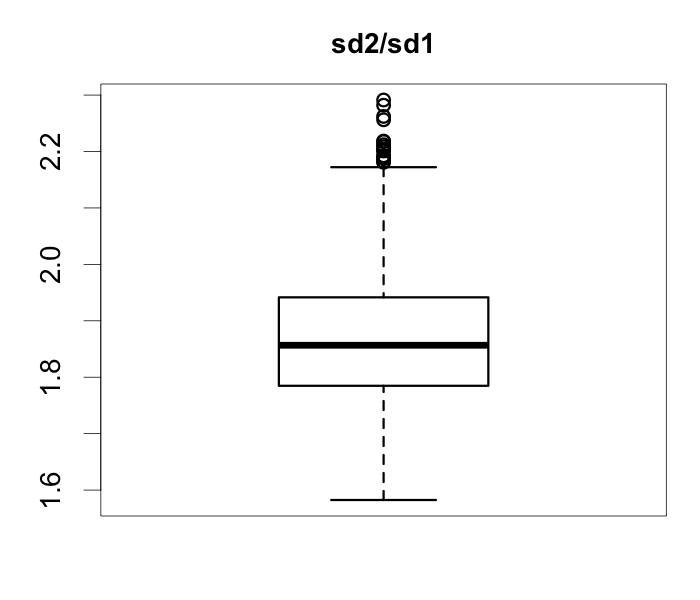}
\end{center}
\caption{Boxplots of the standard deviations of $R$ before and after adding 50 more observations to one sample, as well as their ratio (after/before), from 1,000 simulation runs.}\label{fig:sdR}
\end{figure}

We see that the standard deviation of $R$ after having the additional 50 observations is on average about 1.85 times as large as that for before, which is much larger than one would expect because the squared root of the ratio of the sample sizes is only $\sqrt{150/100}\approx 1.22$.  We see from Figure \ref{fig:Ddiff} that the difference between $R$ and $\bE(R)$ by having the additional observations is not increased as much as the standard deviation in general, resulting in the decrease in the $z$-score, $(\bE(R)-R)/sd(R)$ (see Figure \ref{fig:Zdiff}).

\begin{figure}[!htp]
\begin{center}
\includegraphics[width=.48\textwidth]{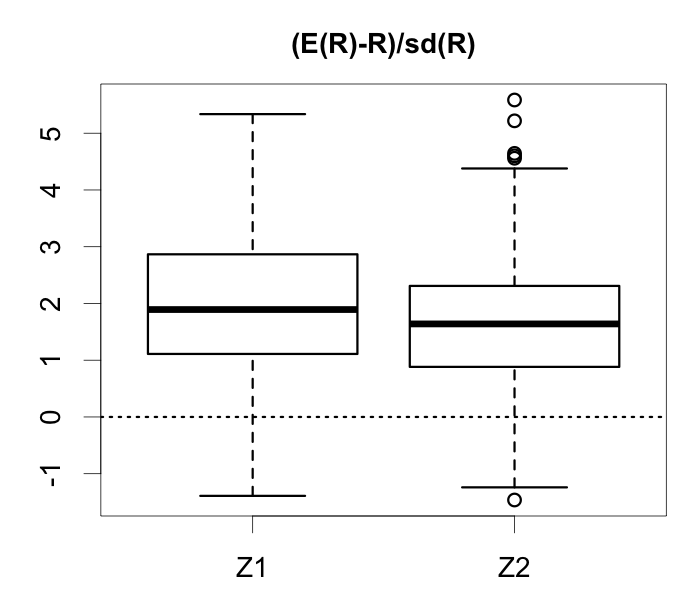}
\includegraphics[width=.48\textwidth]{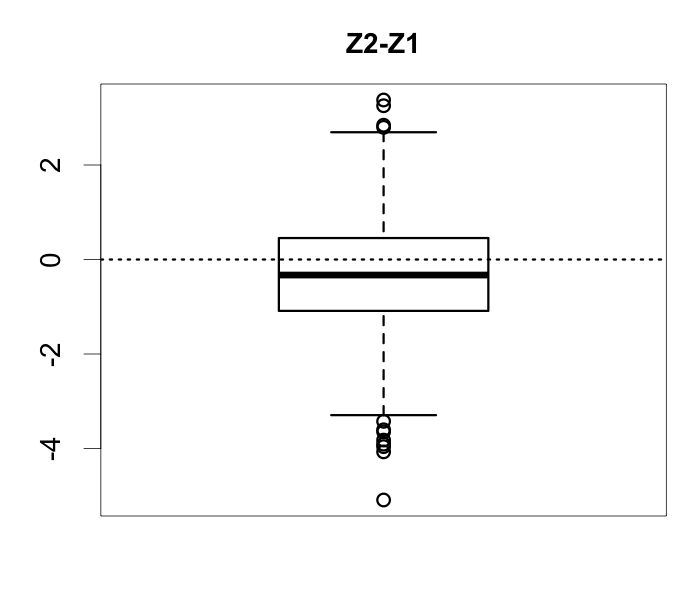}
\end{center}
\caption{Boxplots of the $z$-score, $(\bE(R)-R)/sd(R)$, before and after adding 50 more observations, as well as their difference (after $-$ before), from 1,000 simulation runs.  The horizontal dashed line is at level 0.}\label{fig:Zdiff}
\end{figure}

Therefore, the decrease of power in scenario 2 is mainly due to the variance boosting problem under unequal sample sizes.  To be more exact,
from \cite{friedman1979multivariate} and \cite{chen2016new}, the variance of $R$ under the permutation null distribution is
\begin{align}
\bV(R) & =  \frac{mn(m-1)(n-1)}{N(N-1)(N-2)(N-3)}\times \\ 
& \ \left(4|G| + \frac{(n-m)^2-(N-2)}{(m-1)(n-1)}\left(\sum_{i=1}^N|G_i|^2-\frac{4|G|^2}{N} \right) - \frac{8}{N(N-1)}|G|^2 \right), \nonumber
\end{align}
where $G_i$ the subgraph in $G$ that consists of all edge(s) that connect to node $i$.  So $|G_i|$ is the degree of node $i$ in $G$.

From the Cauchy-Schwarz inequality, we know that
$$\sum_{i=1}^N|G_i|^2 \geq \frac{4|G|^2}{N},$$
and the equality holds only when $|G_i|$'s are equal for all $i$'s.  We call a graph to be \emph{flat} if $\sum_{i=1}^N|G_i|^2-\frac{4|G|^2}{N} = o(N)$, i.e., the degrees of the nodes are similar to each other.  When a graph is not flat, such as a $k$-MST, we see from the expression of $\bV(R)$ that
\begin{equation}\label{eq:VRafter}
\frac{(n-m)^2}{(m-1)(n-1)}\left(\sum_{i=1}^N|G_i|^2-\frac{4|G|^2}{N} \right) 
\end{equation}
contributes a major portion of the variance when the difference between $m$ and $n$ is of order $N$, such as $m=N/3$ and $n=2N/3$.

Figure \ref{fig:quantity} plots $\sum_{i=1}^N|G_i|^2-\frac{4|G|^2}{N}$ and $|G|$ for $k$-MSTs constructed on Euclidean distance in a typical run under scenario 2.  We see that $\sum_{i=1}^N|G_i|^2-\frac{4|G|^2}{N}$ is much larger than $|G|$, especially for large $k$.  In scenario 2, $\frac{(n-m)^2}{(m-1)(n-1)} \approx 0.5$.  When $k>3$, the variance of $R$ under scenario 2 is considerably larger than the corresponding case under scenario 1.  

\begin{figure}[!htp]
\begin{center}
\includegraphics[width=.55\textwidth]{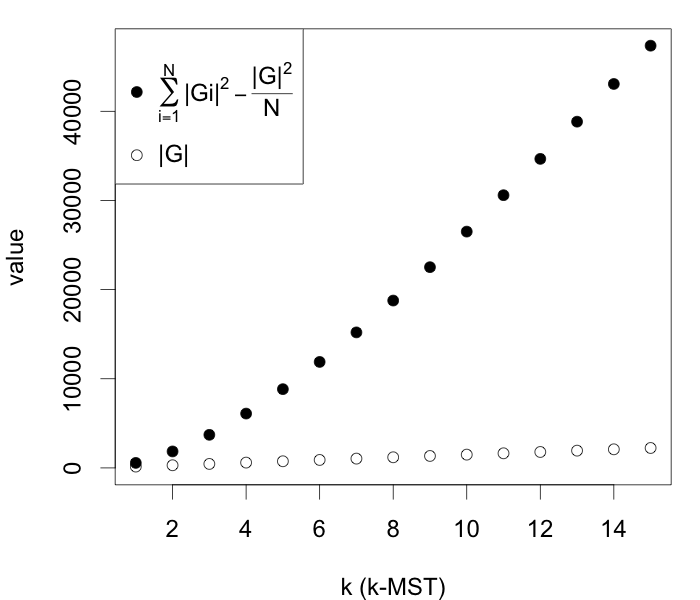}
\end{center}
\caption{The values of $\sum_{i=1}^N|G_i|^2-\frac{4|G|^2}{N}$ and $|G|$ for a $k$-MST constructed on Euclidean distance in a typical run under scenario 2.}\label{fig:quantity}
\end{figure}

To avoid the variance boosting problem, we may choose flat graphs, such as $k$-MDPs, when the sample sizes of the two samples are different.  However, this restricts our choices of the similarity graph.  Instead, we would rather have a test that does not have the variance boosting problem and works for general graphs.

\section{Weighted edge-count test}
\label{sec:weighted}

In this section, we seek to construct a test statistic that captures the signal in a similar way to the edge-count test while not having the variance boosting problem under unequal sample sizes for any graph.

The rationale of the edge-count test is that, under some alternatives, the observations from the same sample tend to connect within themselves, so the number of \emph{between-sample} edges tends to be \emph{less} than its null expectation, or the number of \emph{within-sample} edges tends to be \emph{more} than its null expectation.  Let $R_1$ be the number of edges that connect within the $m$ observations in sample 1 and $R_2$ be that for sample 2.  Then the test statistic for the edge-count test is equivalent to $R_1+R_2$.  This test statistic is not ideal under unequal sample sizes as it treats the two samples in the same way.  However, forming an edge within the sample with a smaller sample size is harder than that for the other sample.  We thus propose to weight the within-sample edge counts by the reciprocal of their corresponding sample sizes.  In particular, we consider the following test statistic:
\begin{align}
R_w & = q R_1 + p R_2, \quad p=\frac{m}{N},\ q=1-p. \label{eq:Rw}
\end{align}
When the test statistic is defined in this way, its variance is well controlled no matter how different $m$ and $n$ are.

\begin{theorem}\label{thm:VRw}
When $m, n\geq 2$,  we have
\begin{equation}\label{eq:bound}
\bV(R_w)\leq \frac{mn(m-1)(n-1)}{N(N-1)(N-2)(N-3)}|G|.
\end{equation}
\end{theorem}

\begin{proof}
Following the expressions of $\bV(R_1)$, $\bV(R_2)$, $\bCov(R_1,R_2)$ in the proof of Theorem 3.2 in \cite{chen2016new}, we have
\begin{align}
\bV(R_w) & = q^2\,\bV(R_1)+p^2\,\bV(R_2)+2\,p\,q\,\bCov(R_1,R_2) \nonumber \\
& = \frac{mn(m-1)(n-1)}{N(N-1)(N-2)(N-3)}\times \\ 
& \quad \left(|G|-\frac{mnN-2m^2-2n^2+2mn}{N^2(m-1)(n-1)} \left(\sum_{i=1}^N|G_i|^2-\frac{4|G|^2}{N} \right) - \frac{2}{N(N-1)}|G|^2 \right). \nonumber
\end{align}
Since $\sum_{i=1}^N|G_i|^2 \geq \frac{4|G|^2}{N}$ from the Cauchy-Schwarz inequality, and $mnN-2m^2-2n^2+2mn = m^2(n-2)+n^2(m-2)+2mn\geq 0$ when $m,n\geq 2$, the result follows.

\end{proof}

A variant of $R_w$ with $q$ and $p$ replaced by $\tilde{q}=(n-1)/(N-2)$ and $\tilde{p}=(m-1)/(N-2)$, respectively, which we denote as $\tilde{R}_w = \tilde{q} R_1 + \tilde{p} R_2$, has the smallest variance among all test statistics of the form $aR_1 + (1-a)R_2$.

\begin{theorem}\label{thm:VRwmin}
For all test statistics of the form $aR_1+bR_2, \ a+b=1$, we have
$$\bV(aR_1+bR_2)\geq \bV(\tilde{R}_w).$$
\end{theorem}
\begin{proof}
Since
\begin{align*}
\bV&(aR_1+bR_2)  = a^2\,\bV(R_1)+b^2\,\bV(R_2)+2\,a\,b\,\bCov(R_1,R_2) \nonumber \\
& = \frac{mn(m-1)(n-1)}{N(N-1)(N-2)(N-3)}\times \\ 
& \quad \left(|G|+\left( a^2\frac{m-2}{n-1} + b^2 \frac{n-2}{m-1} -2ab \right) \left(\sum_{i=1}^N|G_i|^2-\frac{4|G|^2}{N} \right) - \frac{2}{N(N-1)}|G|^2 \right),
\end{align*}
and
\begin{align*}
g(a) & = a^2\frac{m-2}{n-1} + b^2 \frac{n-2}{m-1} -2ab \\
& = a^2 \frac{(N-2)(N-3)}{(n-1)(m-1)} - 2a\frac{N-3}{m-1} + \frac{n-2}{m-1}
\end{align*}
is minimized when $a=(n-1)/(N-2)$, the result follows.
\end{proof}

\begin{remark}
When $m=n$, we have $p=q=0.5$, $\tilde{p}=\tilde{q}=0.5$, then $R_w=\tilde{R}_w=(R_1+R_2)/2 = (|G|-R)/2$.  So the tests based on $R_w$, $\tilde{R}_w$ and $R$ are all equivalent under the balanced design.
\end{remark}

\begin{remark}
For the $k$-MDP, when $N$ is even, every node has degree $k$, so
$$2R_1+R=km, \text{ and } 2R_2+R=kn.$$
Then 
\begin{align*}
R_w & = qR_1+pR_2 = \frac{n}{N}\left(\frac{km-R}{2}\right) + \frac{m}{N}\left(\frac{kn-R}{2}\right) = \frac{kmn}{N} -\frac{R}{2}, \\
\tilde{R}_w & = \tilde{q}R_1+\tilde{p}R_2 = \frac{n-1}{N-2}\left(\frac{km-R}{2}\right) + \frac{m-1}{N-2}\left(\frac{kn-R}{2}\right) = \frac{2kmn-kN}{2(N-2)} -\frac{R}{2}.
\end{align*}
So for the $k$-MDP, when $N$ is even, the tests based on $R_w$, $\tilde{R}_w$ and $R$ are all equivalent. 

When $N$ is odd, the degrees of the nodes are not exactly the same.  However, there are at least $N-k$ nodes with degree $k$ and at most $k$ nodes with degree less than $k$, so the tests based on $R_w$, $\tilde{R}_w$, and $R$ are all very similar.  The same argument holds for any flat graph.  This complies with the earlier observation that the variance boosting problem does not exist for flat graphs.  
\end{remark}

 Asymptotically, when $m$ and $n$ are of the same order, and $N\rightarrow\infty$, the tests based on $R_w$ and $\tilde{R}_w$ are the same.  For finite samples, it turns out that, even though $\bV(\tilde{R}_w)$ is slightly smaller than $\bV(R_w)$, the power of the test based on $R_w$ is slightly higher than that on $\tilde{R}_w$ under locational alternatives.  See Appendix \ref{sec:powerRw} for more details of their power comparisons.  Therefore, we focus on $R_w$ in the following.

\section{Power analysis}
\label{sec:power}

In this section, we check if the weighted edge-count test solves the variance boosting problem.  Consider the illustration example in Section \ref{sec:problem}, we added in the comparison the weighted edge-count test ($R_w$) and the generalized edge-count test ($S$) (Figure \ref{fig:power3}).

\begin{figure}[!htp]
\begin{center}
\includegraphics[width=.6\textwidth]{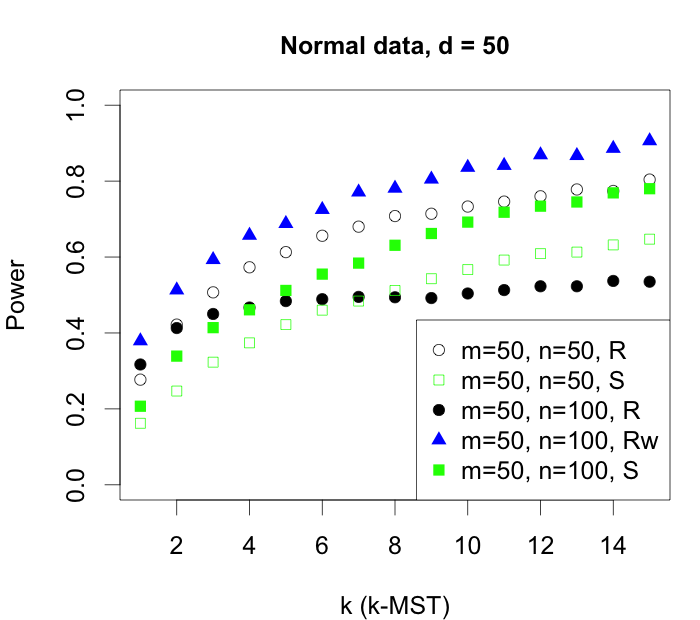}
\end{center}
\caption{The fraction of trials (out of 1,000) that the test rejected the null hypothesis at 0.05 significance level.  }\label{fig:power3}
\end{figure}

We see that the weighted edge-count test (blue triangles) has higher power than the edge-count test when the sample sizes are different for all $k$-MSTs, $k=1,\dots,15$.  It also has higher power than the corresponding smaller sample size scenario ($m=n=50$).  Hence, by controlling the variance, the weighted edge-count test does solve the variance boosting problem in the edge-count test.

Comparing the weighted edge-count test to the generalized edge-count test proposed by \cite{chen2016new}, we see that the weighted edge-count test does have higher power than the generalized edge-count test under locational alternatives, and the weighted edge-count test under scenario 2 ($m=50,n=100$) is the only test that has higher power than the edge-count test under scenario 1 ($m=n=50$) for all $k$-MSTs. 

We checked the performances of the tests under different dimensions (Figure \ref{fig:powernormal}).  The mean differences are chosen so that the tests have moderate power.  We see the same pattern: Under locational alternatives, the weighted edge-count test is more powerful than the edge-count test and the generalized edge-count test when sample sizes are different, and the weighted edge-count test is the only test in scenario 2 that has higher power than the edge-count test in scenario 1 for all $k$-MSTs. 

\begin{figure}[!htp]
\begin{center}
\includegraphics[width=.48\textwidth]{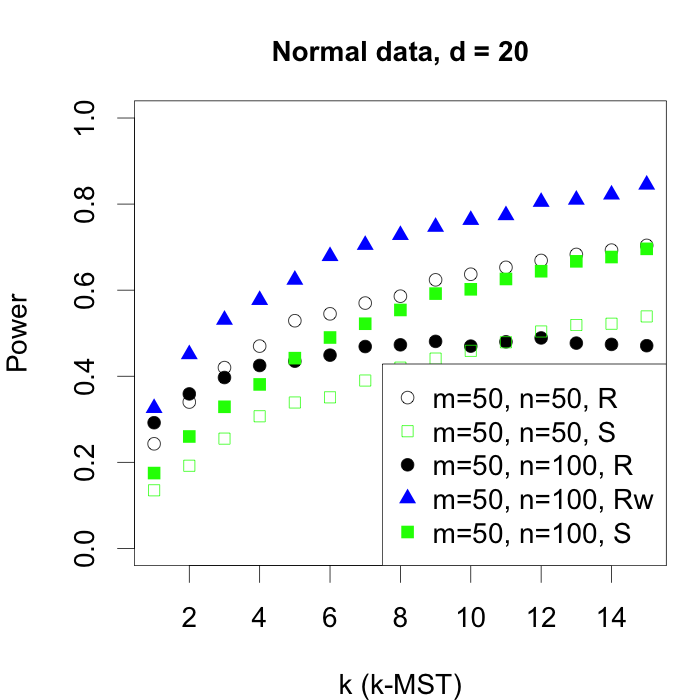}
\includegraphics[width=.48\textwidth]{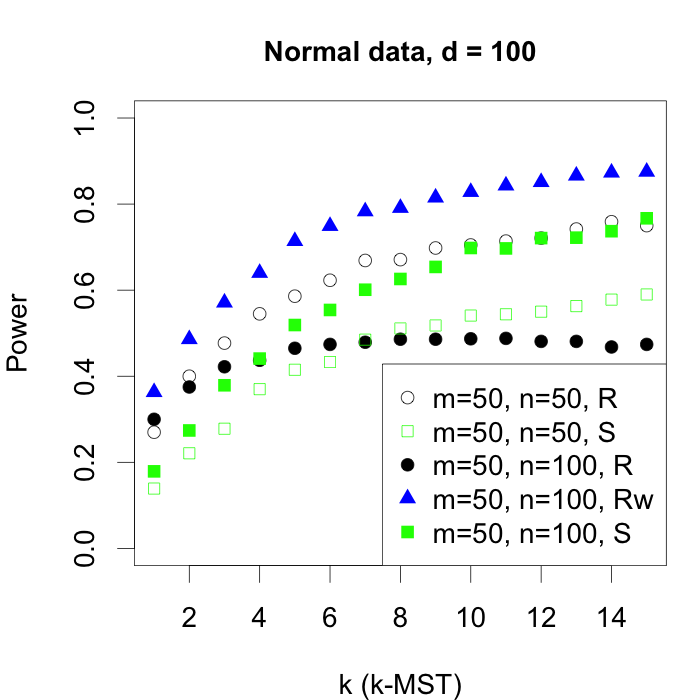}
\end{center}
\caption{The fraction of trials (out of 1,000) that the test rejected the null hypothesis at 0.05 significance level.  The mean differences are of Euclidean distance 1 for $d=20$ and 1.5 for $d=100$.}\label{fig:powernormal}
\end{figure}

We also compared all tests for $t$-distributed data to check how the tests behave when the tail of the distribution is heavier than that of a normal distribution.  The distributions are products of independent $t$ distributions.  The two distributions differ in the mean, and for each dimension, the mean difference is set to be the same as that under the normal case.    The results for $t_{10}$ are shown in Figure \ref{fig:powert10}, and those for $t_5$ are shown in Figure \ref{fig:powert5}.  We see that, overall, all tests have lower power for $t$-distributed data than for normal data, which indicates that the power of all these tests decreases when the tail of the distribution becomes heavier.  Among the three tests, the same pattern retains that the weighted edge-count test outperforms both other tests.

\begin{figure}[!htp]
\begin{center}
\includegraphics[width=.32\textwidth]{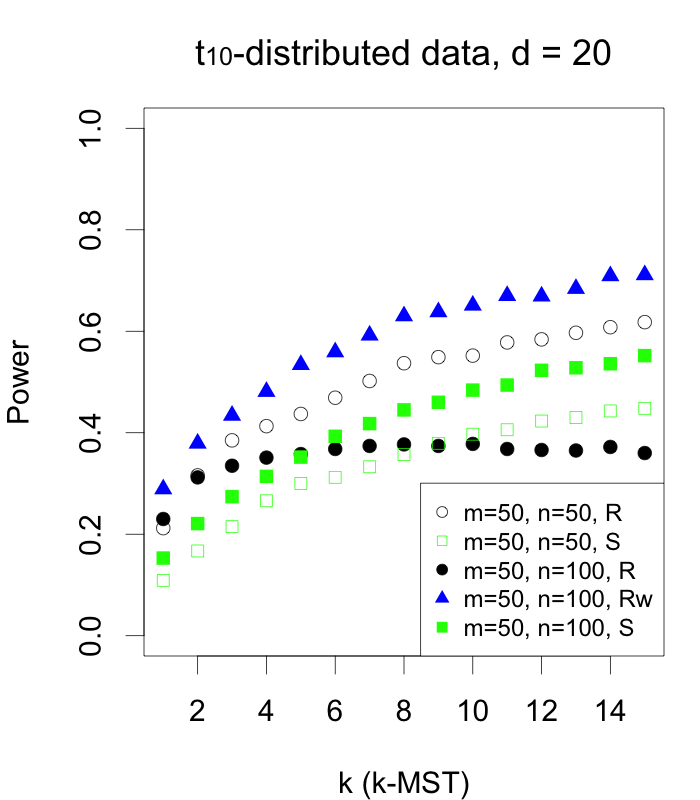}
\includegraphics[width=.32\textwidth]{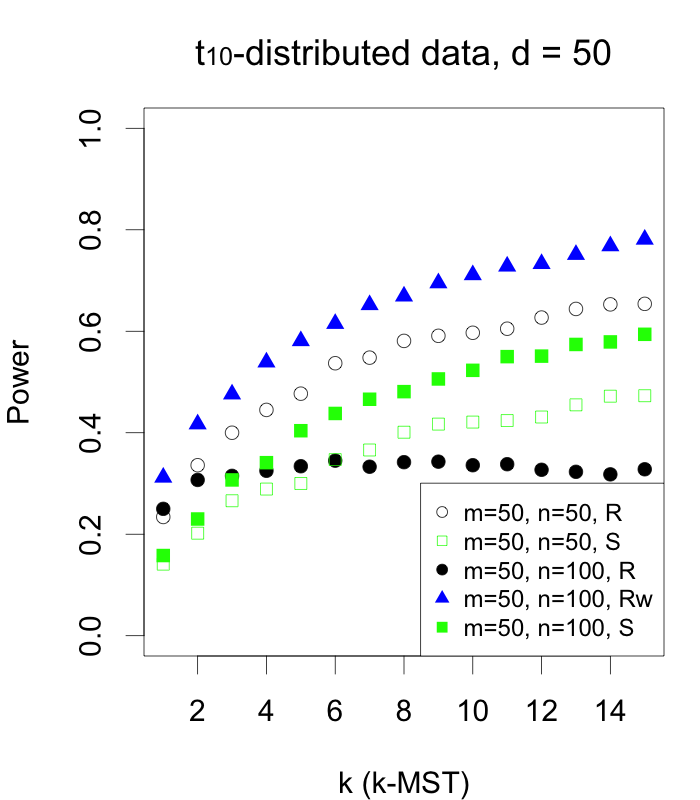}
\includegraphics[width=.32\textwidth]{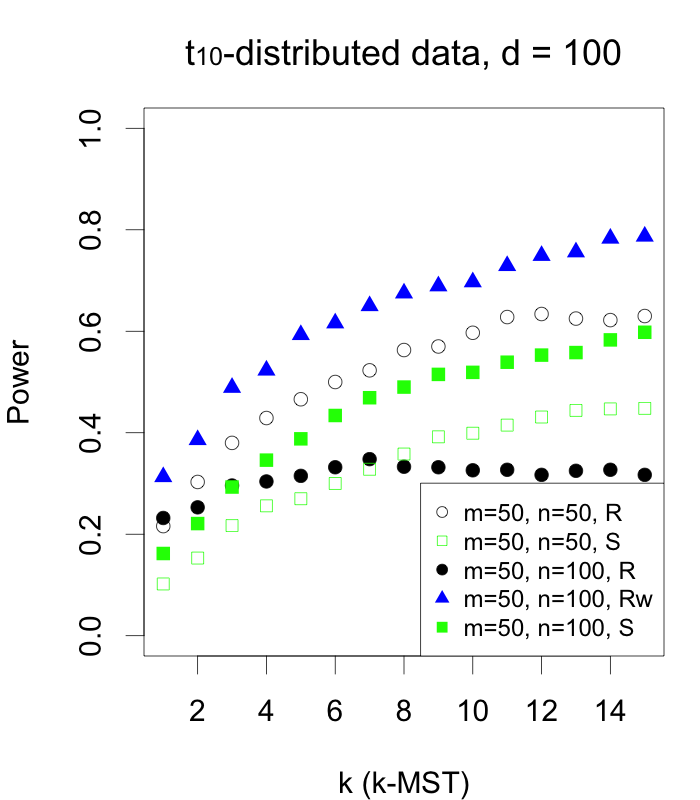}
\end{center}
\caption{The fraction of trials (out of 1,000) that the test rejected the null hypothesis at 0.05 significance level for $t_{10}$-distributed data.}\label{fig:powert10}
\end{figure}

\begin{figure}[!htp]
\begin{center}
\includegraphics[width=.32\textwidth]{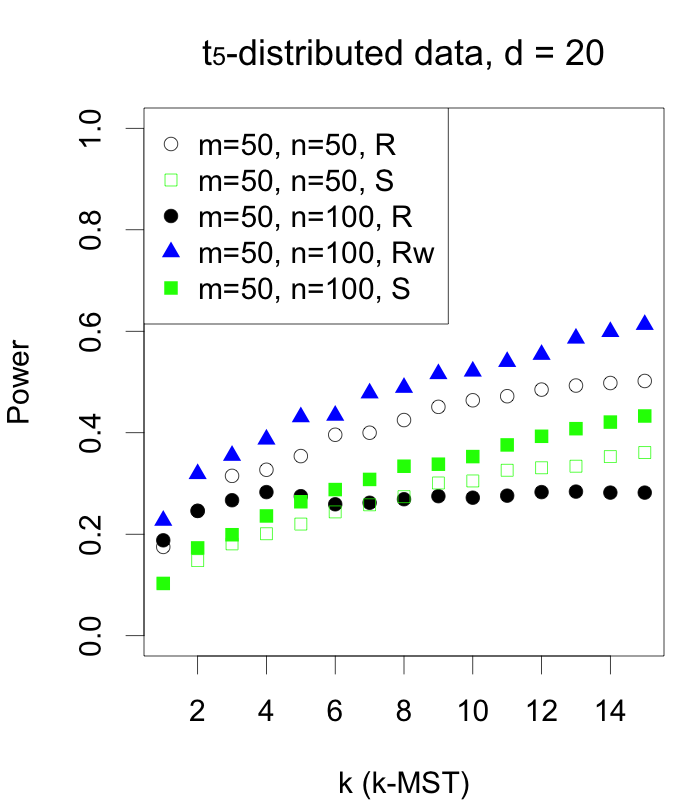}
\includegraphics[width=.32\textwidth]{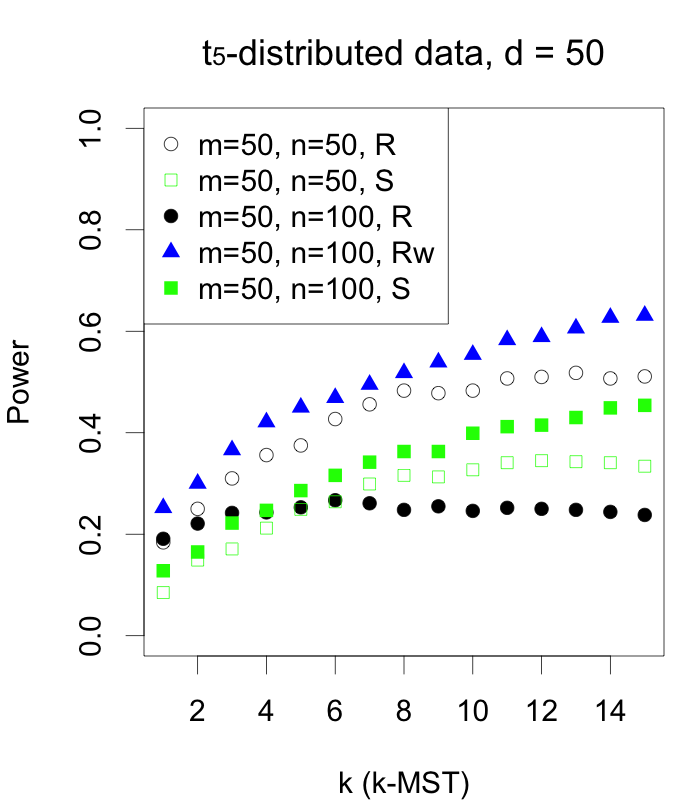}
\includegraphics[width=.32\textwidth]{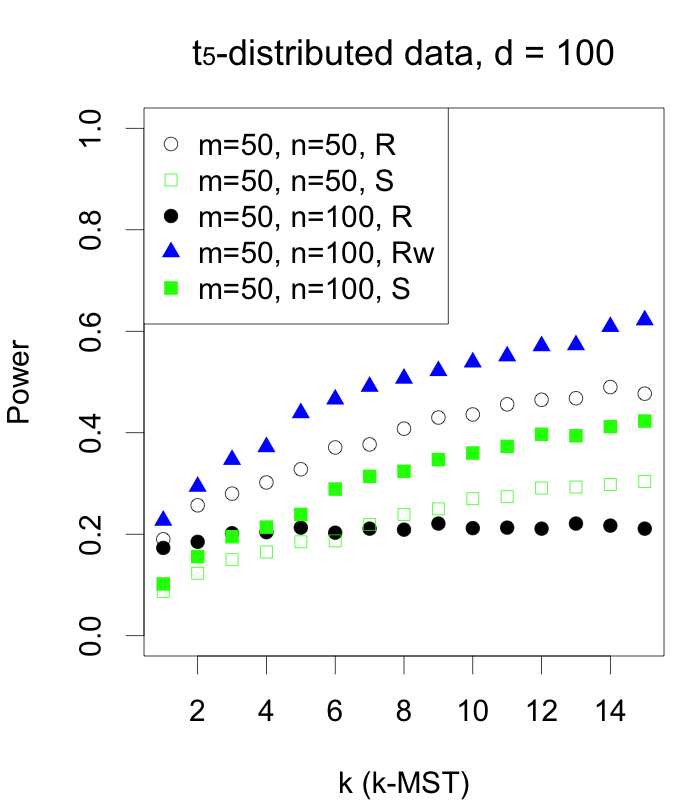}
\end{center}
\caption{The fraction of trials (out of 1,000) that the test rejected the null hypothesis at 0.05 significance level for $t_5$-distributed data.}\label{fig:powert5}
\end{figure}

\section{Asymptotics}
\label{sec:asym}

When the sample size is small, we can obtain the $p$-value of the test directly through permutations.  When the sample size is large, this can be very time consuming.  In the following, we study the asymptotic distribution of $R_w$ under the \emph{usual limiting regime}, which is defined as $|G|,m,n\rightarrow\infty, m/N\rightarrow \lambda\in(0,1)$.  We show that, in the usual limiting regime, $R_w$, normalized by its mean and standard deviation, approaches the standard normal distribution as $N\rightarrow\infty$ under some mild conditions on the similarity graph $G$.  We then check how well the asymptotic null distribution works in approximating $p$-values for finite samples.  In the following, we write $a_n = O(b_n)$ when $a_n$ has the same order as $b_n$, and write $a_n = o(b_n)$ when $a_n$ has order smaller than $b_n$.

\subsection{Asymptotic null distribution}

Before stating the theorem, we define two additional terms on the similarity graph $G$.
 \begin{align}
  A_e & = \{e\} \cup \{e^\prime\in G: e^\prime \text{ and } e \text{ share a node}\},\nonumber \\
  B_e & = A_e \cup \{e^{\prime\prime}\in G: \exists \  e^\prime\in A_e, \text{ such that } e^{\prime\prime} \text{ and } e^\prime \text{ share a node}\}.\nonumber
\end{align}
So $A_e$ is the subgraph in $G$ that consists of all edge(s) that connect to edge $e$, and $B_e$ is the subgraph in $G$ that consists of all edge(s) that connect to any edge in $A_e$.

\begin{theorem}\label{thm:asym}
  If $|G|=O(N^\alpha),\ 1\leq \alpha< 1.5$, $\sum_{e\in G} |A_e||B_e| = o(N^{1.5\alpha})$, and $\sum_{e\in G} |A_e|^2 = o(N^{\alpha+0.5})$, in the usual limiting regime, under the permutation null, 
  \begin{equation}
    \label{eq:asym}
    Z_w := \frac{R_w-\bE(R_w)}{\sqrt{\bV(R_w)}} \overset{\mathcal{D}}{\rightarrow} \mathcal{N}(0,1).
  \end{equation}
\end{theorem}
The proof for this theorem utilizes Stein's method \citep{chen2005stein}.  The complete proof is in Appendix \ref{sec:asymproof}.

\begin{remark}
This theorem also holds for $|G|=O(N^\alpha), \ 0.5<\alpha<1$, along with some other conditions on the graph.  However, such similarity graphs do not connect most of the nodes and are thus not of interest in practice as they missed most of the similarity information among the observations.
\end{remark}

\begin{remark}
The conditions $\sum_{e\in G} |A_e||B_e| = o(N^{1.5\alpha})$ and $\sum_{e\in G} |A_e|^2 = o(N^{\alpha+0.5})$ ensure that the graph does not have a huge hub or a cluster of small hubs, where a hub is a node with a large degree.  If we are only concerned about graphs with $|G|=O(N)$, i.e., $\alpha=1$, then these two conditions degenerate into one condition: $\sum_{e\in G}|A_e||B_e|= o(N^{1.5})$.  Hence, the conditions in Theorem \ref{thm:asym} are much more relaxed than the conditions in obtaining the limiting distribution for the generalized edge-count test statistic in \cite{chen2016new}, which include not only $|G|=O(N)$ and $\sum_{e\in G}|A_e||B_e|= o(N^{1.5})$, but also $\sum_{i=1}^N|G_i|^2=O(N)$ and $\sum_{i=1}^N |G_i|^2 - 4 |G|^2/N = O(N)$. 
\end{remark}

\begin{corollary}
When the graph is a $k$-MST, where $k=O(1)$, based on the Euclidean distance, $Z_w\overset{\mathcal{D}}{\rightarrow} \mathcal{N}(0,1)$ in the usual limiting regime under the null hypothesis.
\end{corollary}  

A $k$-MST, where $k=O(1)$, constructed on the Euclidean distance satisfies all conditions in obtaining the limiting distribution for the generalized edge-count test statistic \citep{chen2016new}.  Hence, it satisfies all conditions for Theorem \ref{thm:asym}, and the result follows.

\subsection{Consistency}

\begin{theorem}\label{thm:consistency}
For two continuous multivariate distributions, if the graph is the MST based on the Euclidean distance, the weighted edge-count test is consistent against all alternatives in the usual limiting regime.
\end{theorem}
This theorem follows straightforwardly from \cite{henze1999multivariate}.  The details are in Appendix \ref{sec:consistencyproof}.  
\begin{remark}
If the graph is a $k$-MST, $k=O(1)$, based on the Euclidean distance, our conjecture is that the same consistency result holds for the weighted edge-count test.  However, a more rigorous proof is needed to extend the almost sure convergences of $R_1/N$ and $R_2/N$ for then 1-MST shown in \cite{henze1999multivariate} to that for the $k$-MST, $k>1$.
\end{remark}

\subsection{Accuracy of the $p$-value approximation from the asymptotics for finite sample sizes}
\label{sec:accuracy}

We here check how large the sample sizes need to be so that the asymptotic $p$-value based on Theorem \ref{thm:asym} is a good approximation to the permutation $p$-value.  Figure \ref{fig:pcheck} shows boxplots of the differences of the two $p$-values (approximated $p$-value from asymptotic results minus $p$-value calculated through 1,000 permutations) from 100 simulation runs under different choices of $m$, $n$, $d$, and the graph.  We see from the boxplots that the approximate $p$-value is very accurate for sample sizes in hundreds.  Making the graph slightly denser, or making the ratio of the two sample sizes higher does not affect the accuracy much. Increasing the dimension does not affect the accuracy much either.

\begin{figure}[!htp]
$d=20$:

\vspace{0.8em}

\includegraphics[width=\textwidth]{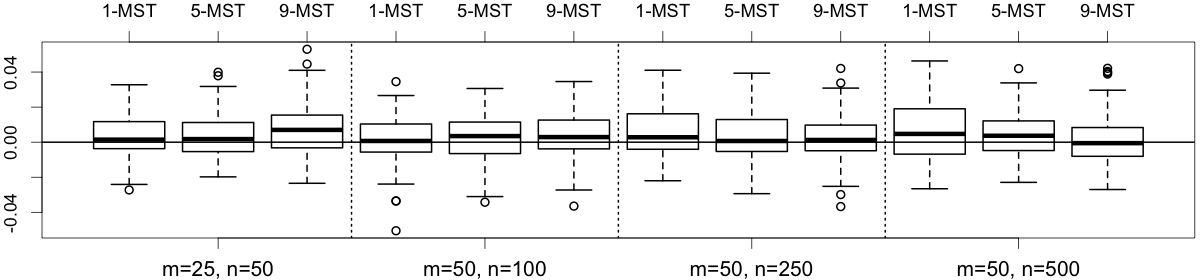} 

\vspace{0.8em}

$d=100$:

\vspace{0.8em}

\includegraphics[width=\textwidth]{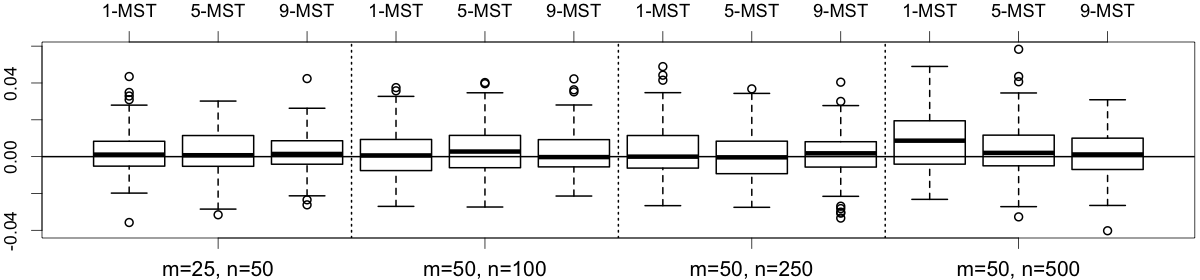}
\caption{Boxplots of the differences between the $p$-value based on asymptotic distribution and the $p$-value calculated directly from 1,000 permutations (100 simulation runs for each setting. $F_1=F_2=\mathcal{N}(\mathbf{0},I_d)$)).} \label{fig:pcheck}
\end{figure}

\section{A real data example}
\label{sec:application}

The MIT Media Laboratory conducted a study following 106 subjects, students and staffs in an institute, who used mobile phones with pre-installed software that can record call logs.  The study lasted from July 2004 to June 2005 \citep{eagle2009inferring}.  The entire dataset is publicly available and can be found at \url{http://realitycommons.media.mit.edu/realitymining.html}. Given the richness of this dataset, many problems can be studied.  One question of interest is whether phone call patterns on weekdays are different from those on weekends.  The phone calls on weekdays and weekends can be viewed as representations of professional relationship and personal relationship, respectively.

We bin the phone calls by day and, for each day, construct a directed phone-call network with the 106 subjects as nodes and a directed edge pointing from person $i$ to person $j$ if person $i$ made at least one call to person $j$ on that day.  We encode the directed network of each day by an adjacency matrix, with 1 for element $[i,j]$ if there is a directed edge pointing from subject $i$ to subject $j$, and 0 if otherwise.  

In this period, there was no call among the subjects on $9.6\%$ of the days in weekdays and $9.3\%$ of the days in weekends.  We remove these days and end up with 214 days in weekdays and 85 days in weekends.  Let $A_1,\dots,A_{214}$ be the adjacency matrices on the 214 weekdays (sample 1) and $A_{215},\dots,A_{299}$ be that on the 85 days in weekends (sample 2).  We consider two distance measures defined as:
\begin{enumerate}[ (1)]
\item the number of different entries: $d(A_i,A_j)=\|A_i-A_j\|_F^2$, where $\|\cdot\|_F$ is the Frobenius norm of a matrix,
\item the number of different entries, normalized by the geometric mean of the total edges in each of the two days: $d(A_i,A_j)=\frac{\|A_i-A_j\|_F^2}{\|A_i\|_F \|A_j\|_F}$.
\end{enumerate}

\begin{figure}[!htp]
\begin{center}
\includegraphics[width=.47\textwidth]{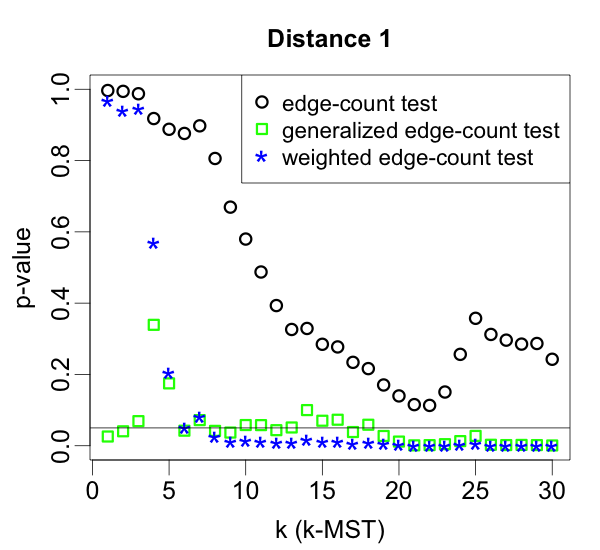} \ \
\includegraphics[width=.47\textwidth]{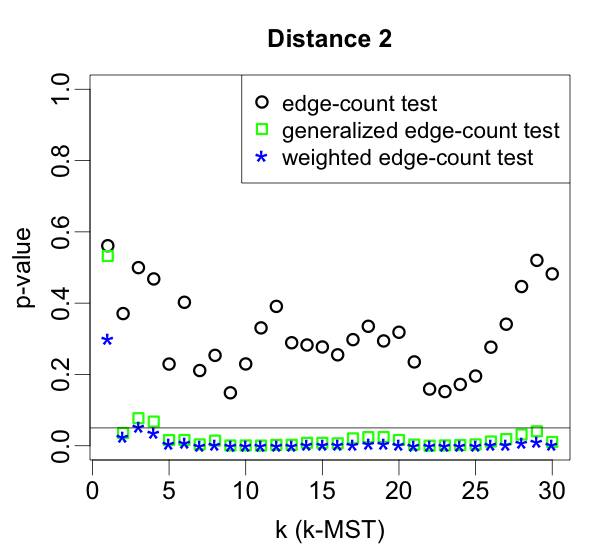}
\caption{The $p$-values of the edge-count tests (circles), the generalized edge-count test (squares), and the weighted edge-count test (stars) on $k$-MSTs constructed on distance 1 and distance 2, respectively.  The horizontal line is of level 0.05.}\label{fig:network}
\end{center}
\end{figure}

Figure \ref{fig:network} shows the $p$-values of the weighted edge-count test, the edge-count test, and the generalized edge-count test, on $k$-MSTs constructed on each of the two distances.
We see that, except for the small $k$'s, the weighted edge-count test rejects the null hypothesis of equal distribution at 0.05 significance level under both distances, while the edge-count test rejects none of them.  The generalized edge-count test rejects most of the scenarios.  Under distance 1, the $p$-values of the generalized edge-count test scatter around 0.05 for $k=8,\dots,20$.

Since all tests have higher power when the graph is slightly denser, the weighted edge-count test rejects the null hypothesis of equal distribution under both distances, the edge-count test does not reject the null under either distance, and the generalized edge-count test more or less rejects the null under both distances.

In the following, we study several representative cases in more detail: the 9-MST and 15-MST under distance 1 and the 9-MST under distance 2.  Table \ref{tab:network} lists the summary statistics of the tests, in particular, the values of $R$, $R_1$, $R_2$, $(R_1+R_2)/2$ and $R_w$, as well as their expectations (mean), standard deviations (SD), and $z$-scores ((value-mean)/SD).  The test based on $(R_1+R_2)/2 = (|G|-R)/2$ is equivalent to that based on $R$.  We include $(R_1+R_2)/2$ in the table to make the comparison between the (unweighted) edge-count test and the weighted edge-count test easier.

\begin{table}[!htp]
\caption{Summary statistics for 9-MST and 15-MST under distance 1 and 9-MST under distance 2.  Sample 1: phone-call networks on weekdays; sample 2: phone-call networks on weekends.}\label{tab:network}
\begin{center}

9-MST, distance 1 \\

\vspace{0.5em}

\begin{tabular}{r||rrrrr}
\hline \hline
& Value & Mean & Value $-$ Mean & SD & $z$-score \\ \hline \hline
$R$ & 1124 & 1095.05 & 28.95 & 66.09 & 0.438 \\ \hline
$R_1$ & 1274 & 1372.03 & -98.03 & 104.95 & -0.934 \\ \hline
$R_2$ & 284 & 214.92 & 69.08 & 42.30 & 1.633 \\ \hline
$\frac{R_1 + R_2}{2}$ & 779 & 793.47 & -14.47 & 33.04 & -0.438 \\ \hline
$R_w$ & 565.44 & 543.86 & 21.58 & 9.50 & 2.272 \\ \hline \hline
\end{tabular}

\vspace{1.5em}

15-MST, distance 1 \\

\vspace{0.5em}

\begin{tabular}{r||rrrrr}
\hline \hline
& Value & Mean & Value $-$ Mean & SD & $z$-score \\ \hline \hline
$R$ & 1770 & 1825.08 & -55.08 & 96.77 & -0.569 \\ \hline
$R_1$ & 2316 & 2286.72 & 29.28 & 155.54 & 0.188 \\ \hline
$R_2$ & 384 & 358.19 & 25.81 & 62.26 & 0.414 \\ \hline
$\frac{R_1 + R_2}{2}$ & 1350 & 1322.46 & 27.54 & 48.38 & 0.569 \\ \hline
$R_w$ & 933.23 & 906.44 & 26.79 & 11.62 & 2.305 \\ \hline \hline
\end{tabular}

\vspace{1.5em}

9-MST, distance 2 \\

\vspace{0.5em}

\begin{tabular}{r||rrrrr}
\hline \hline
& Value & Mean & Value $-$ Mean & SD & $z$-score \\ \hline \hline
$R$ & 1055 & 1095.05 & -40.05 & 38.41 & -1.043 \\ \hline
$R_1$ & 1354 & 1372.03 & -18.03 & 54.99 & -0.327 \\ \hline
$R_2$ & 273 & 214.92 & 58.08 & 23.57 & 2.465 \\ \hline
$\frac{R_1 + R_2}{2}$ & 813.5 & 793.47 & 20.03 & 19.21 & 1.043 \\ \hline
$R_w$ & 580.31 & 543.86 & 36.44 & 10.04 & 3.629 \\ \hline \hline
\end{tabular}

\end{center}
\end{table} 

First, we take a close look at results on the 9-MST based on distance 1.   There are less-than-expectation within-sample edges for the weekday sample and more-than-expectation within-sample edges for the weekend sample.  The sample size for the weekday sample is about 2.5 times as large as the the sample size for the weekend sample.  If we simply add the within-sample edges from the two samples, the number of total within-sample edges is less than its null expectation, falsely indicating that the observations are less likely to form edges within the samples, which leads to the conclusion that the two samples are well mixed and the null hypotheses is not rejected.  This is what the (unweighted) edge-count test does.  

On the other hand, even though $R_1$ deviating from its expectation under the null ($\bE(R_1)-R_1=98.03$) more than $R_2$ deviating from its expectation under the null ($R_2-\bE(R_2)=69.08$), 
it is still more likely to observe values that are more extreme than $R_1$ than that for $R_2$ if we take into account the sample sizes.  The weighted test statistic correctly summarizes the signals provided by both samples and results in a more-than-expectation weighted within-sample edges.  Together with the variance minimizing effect in the weighted edge-count test, the test statistic is significant enough to reject the null hypothesis.

A similar argument holds for the results on the 9-MST based on distance 2.  However, in this case,  $\bE(R_1)-R_1$ is smaller than $R_2-\bE(R_2)$, so the edge-count test also concludes that there are more-than-expectation within-sample edges, while the test is not done effectively that the null is not rejected.  

The case of a 15-MST based on distance 1 is slightly different from the other two cases that the within-sample edges for both samples are larger than their null expectations.  However, the total difference by plain addition is not significant as it is only about half of its corresponding standard deviation.  The variance minimizing effect of the weighted edge-count test is well reflected here that $R_w$ is of a similar amount away from its null expectation compared to the unweighted version ($(R_1+R_2)/2)$), but the standard deviation of $R_w$ is much smaller than that for $(R_1+R_2)/2$, vastly improving the power to detect the signal.  In this case, since both samples are more likely to connect within themselves, the alternative falls in the area that the weighted edge-count test is the most effective.  We see that the generalized edge-count test is not powerful enough for this case, and only the weighted edge-count test rejects the null hypothesis at 0.05 significance level among the three tests.

\section{Relation between the weighted edge-count test and the generalized edge-count test}
\label{sec:relation}
The test statistic of the generalized edge-count test is
\begin{align}\label{eq:S}
S = (R_1-\bE(R_1), R_2-\bE(R_2))\Sigma_R^{-1}\left(\begin{array}{c}R_1-\bE(R_1) \\ R_2-\bE(R_2) \end{array}  \right),
\end{align}
where $\Sigma_R$ is the covariance matrix of the vector $(R_1,R_2)^\prime$.

According to Remark 3.4 in \cite{chen2016new}, when $\sum_{i=1}^N|G_i|^2-\frac{4|G|^2}{N} = O(|G|)$, which is commonly achieved for $k$-MST, $k=O(1)$,
\begin{align}
\lim_{N\rightarrow\infty} \frac{\Sigma_R}{|G|} & = p^2 q^2 \left( \begin{array}{cc}1+rp/q & 1-r \\ 1-r & 1+ r q/p  \end{array} \right),
\end{align}
where $r=\lim_{N\rightarrow\infty}\sum_{i=1}^N(|G_i|^2 - 4|G|^2/N)/|G|$. 

Since
\begin{align*}
\left( \begin{array}{cc}1+rp/q & 1-r \\ 1-r & 1+ r q/p  \end{array} \right) = \left( \begin{array}{cc} p & 1 \\ -q & 1  \end{array} \right)\left( \begin{array}{cc} \frac{r}{pq} & 0 \\ 0 & 1  \end{array} \right) \left( \begin{array}{cc} p & -q \\ 1 & 1  \end{array} \right),
\end{align*}
we have
\begin{align*}
\left( \begin{array}{cc}1+rp/q & 1-r \\ 1-r & 1+ r q/p  \end{array} \right)^{-1} = \left( \begin{array}{cc} 1 & q \\ -1 & p  \end{array} \right)\left( \begin{array}{cc} \frac{pq}{r} & 0 \\ 0 & 1  \end{array} \right) \left( \begin{array}{cc} 1 & -1 \\ q & p  \end{array} \right),
\end{align*}
and 
\begin{align*}
\lim_{N\rightarrow\infty}S & = \frac{(qR_1+pR_2-(q\bE(R_1)+p\bE(R_2)))^2}{p^2q^2|G|} + \frac{(R_1-R_2-(\bE(R_1)-\bE(R_2)))^2}{rpq|G|}.
\end{align*}

We see that the first part of the summation corresponds to the weighted edge-count test.  As discussed in \cite{chen2016new}, under the alternative hypothesis, it can be that (i) both samples are more likely to connect within themselves; or (ii) one sample is more likely to connect within itself while the other sample is less likely to connect within itself.  The two parts in the summation correspond to the two scenarios.  So the generalized edge-count test works under a wider range of alternatives than the weighted edge-count test does, while if the alternative is of scenario (i), the weighted edge-count test is slightly more powerful.  

Since no alternative would lead to the scenario that both samples are less likely to connect within themselves when testing two randomly drawn samples, to make the generalized edge-count test slightly more powerful for both scenarios, we can use the test statistic $\max(Z_w,|Z_\text{diff}|)$ with 
\begin{align*}
Z_w & = \frac{qR_1+pR_2-(q\bE(R_1)+p\bE(R_2))}{pq\sqrt{|G|}}, \\
Z_\text{diff} & = \frac{R_1-R_2-(\bE(R_1)-\bE(R_2))}{\sqrt{rpq|G|}}.
\end{align*}
It can be shown that $(Z_w,Z_\text{diff})^\prime$ is asymptotically bivariate Gaussian distributed and $Z_w$ and $Z_\text{diff}$ are asymptotically independent, so the asymptotic critical value for this variant generalized edge-count test can be easily determined.

\section{Conclusion}
\label{sec:conclusion}

We propose a new two-sample test that utilizes the similarity information among the observations.  In particular, the test is based on a similarity graph constructed on the pooled observations.  Thus, the test can be applied to multivariate data and non-Euclidean data as long as an informative similarity measure on the sample space can be defined.  The classic test of this type, the edge-count test, has an issue when the sample sizes are different.  The new test solves this problem by giving weights based on sample sizes to the different components of the edge-count test statistic.  This weighted edge-count test exhibits substantial power gains in simulation studies. 

The weighted edge-count test statistic, standardized by its mean and standard deviation, is shown to converge to the standard normal distribution under some mild conditions on the similarity graph.  The approximated $p$-value based on the asymptotic results is reasonably accurate to the permutation $p$-value for sample sizes in hundreds, facilitating the application of the test to large datasets.

The weighted edge-count test and the generalized edge-count test in \cite{chen2016new} can be used in a complementary way.  The generalized edge-count test works for a wider range of alternatives under practical sample sizes compared to the weighted edge-count test.  However, under locational alternatives, the power of the weighted edge-count test is higher than that of the generalized edge-count test.  The choice of whether to use the weighted edge-count test or the generalized edge-count test can be done in a similar manner to the choice between the Hotelling $T^2$ test and the generalized likelihood ratio test not assuming equality of covariance matrices.

\section*{Acknowledgements}
\label{sec:acknowledgements}

Hao Chen is supported in part by NSF award DMS-1513653.

\bibliographystyle{plainnat}
\bibliography{newtest}

\appendix

\section{Power comparison between $R_w$ and $\tilde{R}_w$}
\label{sec:powerRw}

According to the definitions of $R_w$ and $\tilde{R}_w$, the tests based on $R_w$ and $\tilde{R}_w$ are very similar for large $m$ and $n$.  We here check the power of them under small $m$ and $n$.  We consider the testing of two samples with one sample $m=10$ observations from $\mathcal{N}(\mathbf{0},I_d)$ and the other sample $n=20$ observations from $\mathcal{N}(\bmu, I_d)$, $\|\bmu\|_2=2, \ d=50$.  The $p$-values are calculated based on 1,000 permutations and the fraction of trials (out of 100) that the test rejected the null hypothesis at 0.05 significance level is listed in Table \ref{tab:powerRw}.

\begin{table}[!htp]
\caption{The fraction of trials (out of 100) that the test rejected the null hypothesis of equal distribution at 0.05 significance level.}\label{tab:powerRw} 
\begin{center}
\begin{tabular}{c|ccccccccc}
\hline
& 1-MST & 2-MST & 3-MST & 4-MST & 5-MST & 6-MST & 7-MST & 8-MST & 9-MST \\ \hline
$R_w$ & 0.42 & 0.46 & 0.46 & 0.52 & 0.57 & 0.57 & 0.55 & 0.50 & 0.52 \\ \hline
$\tilde{R}_w$ & 0.36 & 0.41 & 0.41 & 0.50 & 0.54 & 0.54 & 0.53 & 0.48 & 0.48 \\ \hline
\end{tabular}
\end{center}
\end{table}

We see that the test based on $R_w$ has slightly higher power than that based on $\tilde{R}_w$ for all $k$-MSTs.  To check in a more detailed level, we calculate the difference in $p$-values (the $p$-value of the test based on $R_w$ minus the $p$-value of the test based on $\tilde{R}_w$) for each trial and the boxplots of the differences for each $k$-MST, $k=1,\dots,9$, are shown in Figure \ref{fig:pvaluediff}.  It is clear that the test based on $R_w$ in general has a smaller $p$-value than the test based on $\tilde{R}_w$.  

\begin{figure}[!htp]
\includegraphics[width=\textwidth]{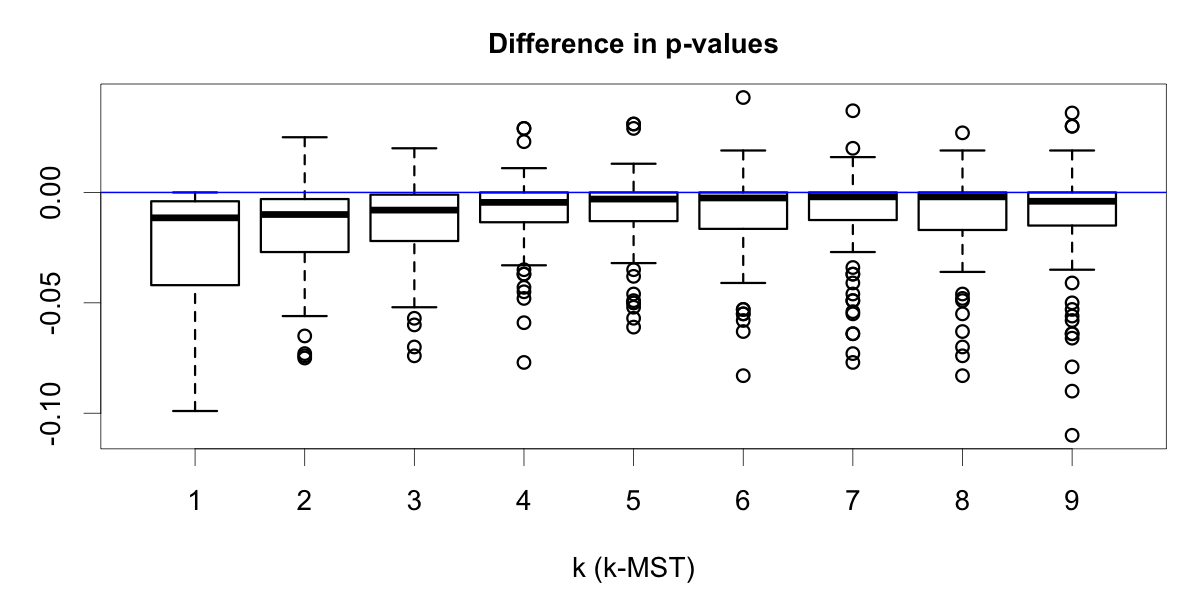}
\caption{Boxplots of the difference between the $p$-values of the tests based on $R_w$ and $\tilde{R}_w$ (the $p$-value of the test based on $R_w$ minus the $p$-value of the test based on $\tilde{R}_w$) for $k$-MST, $k=1,\dots,9$. The horizontal line is at level 0.}\label{fig:pvaluediff}
\end{figure}

According to the above comparison and the simpler form of $R_w$, we recommend to use $R_w$ in practice.

\section{Proof of Theorem \ref{thm:asym}}
\label{sec:asymproof}

The proof of Theorem \ref{thm:asym} relies on Stein's method.  Consider sums of the form $W=\sum_{i\in{\cal J}} \xi_i,$
where $\mathcal{J}$ is an index set and $\xi$ are random variables with $\bE \xi_i=0$, and $\bE (W^2)=1$.  The following assumption restricts the dependence between $\{\xi_i:~i \in \mathcal{J}\}$.
\begin{assumption} \cite[p.\, 17]{chen2005stein}
  \label{assump:LD}
For each $i\in{\cal J}$ there exists $S_i \subset T_i \subset {\cal J}$ such that $\xi_i$ is independent of $\xi_{S_i^c}$ and $\xi_{S_i}$ is independent of $\xi_{T_i^c}$.
\end{assumption}
We will use the following theorem in proving Theorem \ref{thm:asym}.
\begin{theorem}\label{thm:3.4} \cite[Theorem 3.4]{chen2005stein}
Under Assumption \ref{assump:LD}, we have
$$\sup_{h\in Lip(1)} |\bE h(W) - \bE h(Z)| \leq \delta,$$
where $Lip(1) = \{h: \mathbb{R}\rightarrow \mathbb{R}; \|h^\prime\|\leq 1 \}$, $Z$ has ${\cal N}(0,1)$ distribution and
 $$\delta = 2 \sum_{i\in{\cal J}} (\bE|\xi_i \eta_i\theta_i| + |\bE(\xi_i\eta_i)|\bE|\theta_i|) + \sum_{i\in{\cal J}} \bE|\xi_i\eta_i^2|$$
with $\eta_i = \sum_{j\in S_i}\xi_j$ and $\theta_i = \sum_{j\in T_i} \xi_j$, where $S_i$ and $T_i$ are defined in Assumption \ref{assump:LD}.
\end{theorem}

To prove Theorem \ref{thm:asym}, we take one step back to study the statistic under the bootstrap null distribution, which is defined as follows: For each observation, we assign it to be from sample 1 with probability $m/N$, and from sample 2 with probability $n/N$, independently of other observations.  We use $g_i$ to denote the sample assignment for observation $i$, with $g_i=1$ if observation $i$ is assigned to be from sample 1 and $g_i=2$ if otherwise.  Let $X=\sum_{i=1}^N I(g_i=1)$ be the number of observations assigned to be from sample 1, where $I(\cdot)$ is the indicator function.  Then, the bootstrap null distribution conditioning on $X=m$ is equivalent to the permutation null distribution.  We use $\bPB$, $\bEB$, $\bVB$, $\bCovB$ to denote the probability, expectation, variance, and covariance under the bootstrap null distribution, respectively.  (We here add the subscript $_{\mathbf{P}}$ to denote the corresponding quantities under the permutation null distribution.)

Let $p_N=m/N, q_N=1-p_N$, then $\lim_{N\rightarrow\infty} p_N = \lambda, \lim_{N\rightarrow\infty}q_N=1-\lambda$.  For any $e=(e_-,e_+)\in G$, let
\begin{align*}
\beta_e = \left\{ \begin{array}{ll} 1 & \text{ if } g_{e_-}=g_{e_+}=1, \\ 2 & \text{ if } g_{e_-}=g_{e_+}=2, \\ 0 & \text{ if } g_{e_-}\neq g_{e_+}.  \end{array}\right.
\end{align*}

Given that the $g_i$'s are independent under the bootstrap null distribution, we have
\begin{align*}
  \bEB(R_1) & = \frac{m^2}{N^2}|G|, \\
  \bEB(R_2) & = \frac{n^2}{N^2}|G|, \\
  \bVB(R_1) & = \frac{m^2 n^2}{N^4}|G| + \frac{m^3 n}{N^4}\sum_{i=1}^N|G_i|^2, \\
  \bVB(R_2) & = \frac{m^2 n^2}{N^4}|G| + \frac{m n^3}{N^4}\sum_{i=1}^N|G_i|^2, \\
  \bCovB(R_1,R_2) & = \frac{m^2 n^2}{N^4}|G| - \frac{m^2 n^2}{N^4}\sum_{i=1}^N|G_i|^2.
\end{align*}
For $R_w = q_N R_1 + p_N R_2$, we have
\begin{align*}
\bEB(R_w) & = \frac{mn}{N^2}|G| :=\mu_B, \\
\bVB(R_w) & = \frac{m^2 n^2}{N^4}|G| :=\sigma_B^2.
\end{align*}
In contrast, we have
\begin{align*}
\bE_\bP(R_w) & = \frac{mn(N-2)}{N^2(N-1)}|G| :=\mu_P, \\
\bV_\bP(R_w) & \frac{mn(m-1)(n-1)}{N(N-1)(N-2)(N-3)}\times \\ 
& \quad \left(|G|-\frac{mnN-2m^2-2n^2+2mn}{N^2(m-1)(n-1)} \left(\sum_{i=1}^N|G_i|^2-\frac{4|G|^2}{N} \right) - \frac{2}{N(N-1)}|G|^2 \right) \\
& :=\sigma_P^2,
\end{align*}

Let 
\begin{align*}
  Z_w^B & = \frac{R_w-\mu_B}{\sigma_B}, \quad 
  Z_X^B = \frac{X-m}{\sqrt{mn/N}}.
\end{align*}

Under the conditions of Theorem \ref{thm:asym}, as $N\rightarrow\infty$, we can prove the following results:
\begin{enumerate}[(i)]
\item $(Z_w^B, Z_X^B)$ becomes bivariate Gaussian distributed under the bootstrap null.
\item $$\frac{\sigma_P}{\sigma_B}\rightarrow 1, \quad \frac{\mu_B - \mu_P}{\sigma_P}\rightarrow 0.$$
\end{enumerate}

From (i) and together with $\bVB(Z_X^B)=1$, as $N\rightarrow \infty$, the conditional distribution of $Z_w^B$ given $Z_X^B$ converges to a Gaussian distribution under the bootstrap null distribution.  Since the bootstrap null distribution conditioning on $Z_X^B=0$ is equivalent to the permutation null distribution, $Z_w^B$ follows a Gaussian distribution under the permutation null distribution as $N\rightarrow\infty$.  Notice that 
$$ Z_w = \frac{\sigma_B}{\sigma_P}\left(Z_w^B + \frac{\mu_B - \mu_P}{\sigma_B}\right),$$
together with (ii), we have $Z_w$ converges to a Gaussian distribution under the permutation null distribution as $N\rightarrow\infty$.  

In the following, we prove results (i) and (ii).

To prove (i), by the Cram$\acute{\text{e}}$r-Wold device, we only need to show that $W=a_1 Z_w^B + a_2 Z_X^B$ is asymptotically Gaussian distributed for any combination of $a_1$ and $a_2$ such that $\bVB(W)>0$.

Let 
\begin{align*}
  \xi_e & = \frac{a_1}{\sigma_B}\left(\frac{n}{N}I_{\beta_e=1} + \frac{m}{N}I_{\beta_e=2} -\frac{mn}{N^2} \right), \\
  \xi_i & = a_2 \frac{I_{g_i=1}-m/N}{\sigma_0}, \ \sigma_0=\sqrt{mn/N}.
\end{align*}
Then $W=\sum_{j\in \mathcal{J}} \xi_j$, where $\mathcal{J}=\{e\in G\} \cup \{1,\dots,N\}$.  Let $a=\max(|a_1|,|a_2|)$, then $|\xi_e|\leq a/\sigma_B$ and $|\xi_i|\leq a/\sigma_0$. 

For $e=(e_-,e_+) \in \mathcal{J}$, let 
\begin{align*}
  S_e & = A_e\cup \{e_-,e_+\}, \\
  T_e & = B_e\cup\{\text{nodes in $A_e$}\}.
\end{align*}
Then $S_e$ and $T_e$ satisfy Assumption \ref{assump:LD}.  

For $i\in \{1,\dots,N\}$, let 
\begin{align*}
  S_i & = \{e\in G_i\}\cup \{i\}, \\
  T_i & = \{e\in G_{i,2}\}\cup \{\text{nodes in $G_i$}\},
\end{align*}
where $G_{i,2} = \{(j,l)\in G: j\in G_i\}$ is the subgraph is $G$ that consists of all edges that connect to any node in $G_i$. Then $S_i$ and $T_i$ satisfy Assumption \ref{assump:LD}.

For $j\in \mathcal{J}$, let $\eta_j = \sum_{k\in K_j} \xi_k$, $\theta_j=\sum_{k\in L_j} \xi_k$.
By Theorem \ref{thm:3.4}, we have $\sup_{h\in Lip(1)} |\bEB h(W) - \bE h(Z)|\leq \delta$ for $Z\sim \mathcal{N}(0,1)$, where
\begin{align*}
  \delta & = \frac{1}{\sqrt{\bVB(W)}}\left(2\sum_{j\in\mathcal{J}}(\bEB|\xi_j\eta_j\theta_j| + |\bEB(\xi_j\eta_j)|\bEB|\theta_j|) + \sum_{j\in\mathcal{J}} \bEB|\xi_j\eta_j^2| \right) \\
  & \leq \frac{a^3}{\sqrt{\bVB(W)}} \left( 5 \sum_{e\in G} \frac{1}{\sigma_B}\left(\frac{|A_e|}{\sigma_B}+\frac{2}{\sigma_0}\right)\left(\frac{|B_e|}{\sigma_B}+\frac{2|A_e|}{\sigma_0}\right) \right.\\
  & \quad \quad \quad \quad \quad \quad \quad \quad \quad  +\left. 5 \sum_{i=1}^N \frac{1}{\sigma_0}\left(\frac{|G_i|}{\sigma_B}+\frac{1}{\sigma_0}\right)\left(\frac{|G_{i,2}|}{\sigma_B}+\frac{2|G_i|}{\sigma_0}\right)  \right).
\end{align*}
Since $\sigma_B=O(|G|^{0.5})$ and $\sigma_0=O(N^{0.5})$, then as long as 
\begin{align}
\sum_{e\in G}|A_e||B_e| & = o(|G|^{1.5}), \label{eq:e1} \\
\sum_{e\in G}|A_e|^2 & = o(|G|\cdot N^{0.5}), \label{eq:e2} \\
\sum_{e\in G}|B_e| & = o(|G|\cdot N^{0.5}), \label{eq:e3} \\
\sum_{e\in G}|A_e| & = o(|G|^{0.5}\cdot N), \label{eq:e4} \\
\sum_{i=1}^N |G_i||G_{i,2}| & = o(|G|\cdot N^{0.5}), \label{eq:e5} \\
\sum_{i=1}^N |G_i|^2 & = o(|G|^{0.5}\cdot N), \label{eq:e6} \\
\sum_{i=1}^N |G_{i,2}| & = o(|G|^{0.5}\cdot N), \label{eq:e7} \\
\sum_{i=1}^N |G_i| & = o(N^{1.5}), \label{eq:e8}
\end{align}
result (i) follows.  We next show that the conditions in Theorem \ref{thm:asym} ($|G|=O(N^\alpha), \ 1\leq \alpha<1.5$, $\sum_{e\in G}|A_e||B_e| = o(N^{1.5\alpha})$, $\sum_{e\in G}|A_e|^2 = o(N^{\alpha+0.5})$) are enough to show \eqref{eq:e1}-\eqref{eq:e8}.

First of all, if we substitute $|G|$ by $O(N^\alpha)$, \eqref{eq:e1}, \eqref{eq:e2}, and \eqref{eq:e8} follow immediately.  

Since $|B_e|\leq \sum_{e^\prime \in A_e} |A_{e^\prime}|$, we have $\sum_{e\in G} |B_e|\leq \sum_{e\in G}\sum_{e^\prime \in A_e} |A_{e^\prime}|$.  For the latter quantity, for any $e^*\in G$, $|A_{e^*}|$ appears $|A_{e^*}|$ times because $e^*$ is the first neighbor of $(|A_{e^*}|-1)$ edges.  Hence, the latter quantity equals $\sum_{e\in G} |A_e|^2$.  So \eqref{eq:e2} ensures \eqref{eq:e3}.

By Cauchy-Schwarz, we have $\sum_{e\in G}|A_e| \leq \sqrt{\sum_{e\in G} |A_e|^2 \sum_{e\in G} 1^2} = o(|G|\cdot N^{0.25})$.  Since $(\alpha+0.25)-(0.5\alpha+1) = 0.5\alpha-0.75<0$ as $\alpha<1.5$, \eqref{eq:e4} follows.   

We use $\mathcal{V}_{G_i}$ to denote the vertex set of $G_i$. Since $|G_{i,2}| \leq \sum_{j\in \mathcal{V}_{G_i}} |G_j|$, we have $\sum_{i=1}^N |G_{i,2}| = \sum_{i=1}^N \sum_{j\in \mathcal{V}_{G_i}} |G_j| = \sum_{(i,j)\in G} (|G_i|+|G_j|) \leq \sum_{e\in G} 2|A_e|$. So \eqref{eq:e4} ensures \eqref{eq:e7}.

Also, $\sum_{e\in G}|A_e| = \sum_{(i,j)\in G} (|G_i|+|G_j|-1) = \sum_{i=1}^N \sum_{j\in \mathcal{V}_{G_i}} |G_j| - |G| = \sum_{i=1}^N |G_i|^2 - |G|$ because for each $i^*$, $|G_{i^*}|$ appears $|G_{i^*}|$ times in the summation $\sum_{i=1}^N \sum_{j\in \mathcal{V}_{G_i}} |G_j|$.  Since $|G|=o(N^{0.5\alpha +1})$ when $|G|=O(N^\alpha),\ \alpha<1.5$, \eqref{eq:e4} and \eqref{eq:e6} are equivalent.

For \eqref{eq:e5}, we have $\sum_{i=1}^N |G_i||G_{i,2}| \leq \sum_{i=1}^N |G_i| \sum_{j \in \mathcal{V}_{G_i}} |G_j| = \sum_{i=1}^N \sum_{j\in \mathcal{V}_{G_i}} |G_i||G_j| = 2\sum_{(i,j)\in G} |G_i||G_j| \leq 2\sum_{e\in A_e} |A_e|^2$.  So \eqref{eq:e2} ensures \eqref{eq:e5}.

Hence, all \eqref{eq:e1}-\eqref{eq:e8} can be derived from conditions in Theorem \ref{thm:asym}.

Next we prove result (ii).  We have
\begin{align*}
\lim_{N\rightarrow\infty} \frac{\sigma^2_P}{\sigma^2_B} & = \lim_{N\rightarrow\infty} \frac{|G|-\frac{mnN-2m^2-2n^2+2mn}{N^2(m-1)(n-1)} \left(\sum_{i=1}^N|G_i|^2-\frac{4|G|^2}{N} \right) - \frac{2}{N(N-1)}|G|^2 }{|G|} \\
& = 1-\lim_{N\rightarrow\infty} \left(\frac{\sum_{i=1}^N|G_i|^2}{N|G|} - \frac{4|G|}{N^2} + \frac{2|G|}{N^2} \right).
\end{align*}
Because $|G|=O(N^\alpha),\ \alpha<1.5$ and $\sum_{i=1}^N |G_i|^2 = o(N^{0.5\alpha+1})$ from the proof for result (i), we have $\lim_{N\rightarrow\infty} \left(\frac{\sum_{i=1}^N|G_i|^2}{N|G|} - \frac{4|G|}{N^2} + \frac{2|G|}{N^2} \right)=0$, so $\lim_{N\rightarrow\infty}\frac{\sigma_P}{\sigma_B} = 1$.

Since $\mu_B-\mu_P = \frac{mn}{N^2(N-1)}|G|$, we have
\begin{align}
\lim_{N\rightarrow\infty}\frac{\mu_B-\mu_P}{\sigma_P} = \lim_{N\rightarrow\infty}\frac{\sqrt{|G|}}{N-2} = 0,
\end{align}
when $|G|=O(N^\alpha),\ \alpha<1.5$.

\section{Proof to Theorem \ref{thm:consistency}}
\label{sec:consistencyproof}
Let the density functions of the two multivariate distributions be $f$ and $g$.  When the similarity graph is the MST constructed on the Euclidean distance, from \cite{henze1999multivariate}, we have
$$\frac{R_1}{N} \rightarrow  \int \frac{\lambda^2f^2(x)}{\lambda f(x)+(1-\lambda)g(x)} dx \quad \text{almost surely, and}$$
$$\frac{R_2}{N}\rightarrow  \int \frac{(1-\lambda)^2g^2(x)}{\lambda f(x)+(1-\lambda)g(x)} dx \quad \text{almost surely.}$$
So
\begin{align*}
\frac{R_w-\bE(R_w)}{N} & = \frac{n}{N}\frac{R_1}{N} + \frac{m}{N}\frac{R_2}{N}-\frac{m(m-1)n/N+n(n-1)m/N}{N^2(N-1)}|G| \\
& \quad \rightarrow \lambda(1-\lambda)\left(\int\frac{\lambda f^2(x) + (1-\lambda) g^2(x)}{\lambda f(x) + (1-\lambda) g(x)} dx -1 \right) \quad \text{almost surely}.
\end{align*}
Since
\begin{align*}
 \int\frac{\lambda f^2(x) + (1-\lambda) g^2(x)}{\lambda f(x) + (1-\lambda) g(x)} dx - 1 & = \int\frac{\lambda f(x)(f(x)-g(x))}{\lambda f(x)+(1-\lambda)g(x)} dx \\
 & = \int\frac{(1-\lambda)g(x)(g(x)-f(x))}{\lambda f(x)+(1-\lambda)g(x)} dx,
\end{align*}
we have
\begin{align*}
&\int\frac{\lambda f^2(x) + (1-\lambda) g^2(x)}{\lambda f(x) + (1-\lambda) g(x)} dx - 1 \\
& \quad = (1-\lambda)\int\frac{\lambda f(x)(f(x)-g(x))}{\lambda f(x)+(1-\lambda)g(x)} dx + \lambda\int\frac{(1-\lambda)g(x)(g(x)-f(x))}{\lambda f(x)+(1-\lambda)g(x)} dx \\
& \quad = \lambda(1-\lambda)\int\frac{(f(x)-g(x))^2}{\lambda f(x)+(1-\lambda)g(x)}dx .
\end{align*}
Therefore, $\int\frac{\lambda f^2(x) + (1-\lambda) g^2(x)}{\lambda f(x) + (1-\lambda) g(x)} dx \geq 1$, and it is strictly greater than 1 if $f$ and $g$ differ on a set of positive measure.  Since $\sqrt{\bV(R_w)}=O(\sqrt{N})$ according to Theorem \ref{thm:VRw}, the test is consistent.

\end{document}